\documentclass[11pt]{article}

\usepackage[utf8]{inputenc}
\usepackage[T1]{fontenc}

\usepackage[margin=1in]{geometry}

\usepackage{amsmath,amssymb,amsthm}

\usepackage{graphicx}
\usepackage{float}

\usepackage{booktabs}

\usepackage{algorithm}
\usepackage{algpseudocode}

\usepackage{natbib}
\usepackage[colorlinks=true,linkcolor=blue,citecolor=blue,urlcolor=blue]{hyperref}
\usepackage{url}

\usepackage{enumitem}
\setlist{nosep,leftmargin=*}
\theoremstyle{definition}
\newtheorem{defn}{Definition}
\newtheorem{prop}[defn]{Proposition}
\newtheorem*{remark}{Remark}

\title{Lean Atlas: An Integrated Proof Environment\\for Scalable Human-AI Collaborative Formalization}

\author{
  Banri Yanahama\textsuperscript{1}\thanks{\texttt{banri.yanahama@nyx.foundation}} \quad
  Akiyoshi Sannai\textsuperscript{2,1,3}\thanks{\texttt{sannai.akiyoshi.7z@kyoto-u.ac.jp}}
  \\[6pt]
  \textsuperscript{1}Nyx Foundation \quad
  \textsuperscript{2}Kyoto University\\[3pt]
  \textsuperscript{3}Large Language Model Research and Development Center, National Institute of Informatics
}

\date{}

\begin{document}

\maketitle
\vspace{-2em}

\begin{abstract}
\small\noindent
AI-driven autoformalization of mathematics is advancing rapidly. However, the type checker of a proof assistant guarantees only the logical correctness of proofs; it does not verify whether propositions and definitions faithfully capture their intended mathematical content. Consequently, AI-generated formal proofs can exhibit semantic hallucination---passing the type checker yet failing to express the intended mathematics. We propose a human-in-the-loop approach in which human scientists and AI collaboratively produce formal proofs, with humans responsible for the semantic verification of propositions and definitions. To realize this approach, we develop Lean Atlas, a Lean~4 tool that visualizes the dependency graph of a Lean~4 project as an interactive web viewer, enabling human scientists to grasp the overall structure of a formalization efficiently. Its core feature, Lean Compass, is an algorithm that, given a selected theorem set, automatically extracts the project-specific nodes whose semantic correctness can affect those target statements, thereby reducing the candidate set for semantic review in large-scale formalizations. We further define \emph{aligned Lean code} as formalization code that has undergone human semantic verification, and propose it as a quality standard for AI-generated formalizations. We evaluate the tool on six Lean~4 formalization projects with different structural characteristics; proof-heavy projects (PrimeNumberTheoremAnd, Carleson, Brownian Motion) achieved 94--99\% average node reduction, a 6-theorem milestone subset of FLT achieved 59.8\%, mixed PhysLib 69.0\%, and definition-heavy XMSS 27.3\%. Lean Atlas is available as \href{https://github.com/NyxFoundation/lean-atlas}{open-source software}.
\end{abstract}
\vspace{-1em}

\begin{figure}[H]
  \centering
  \includegraphics[width=0.75\textwidth,height=0.25\textheight,keepaspectratio]{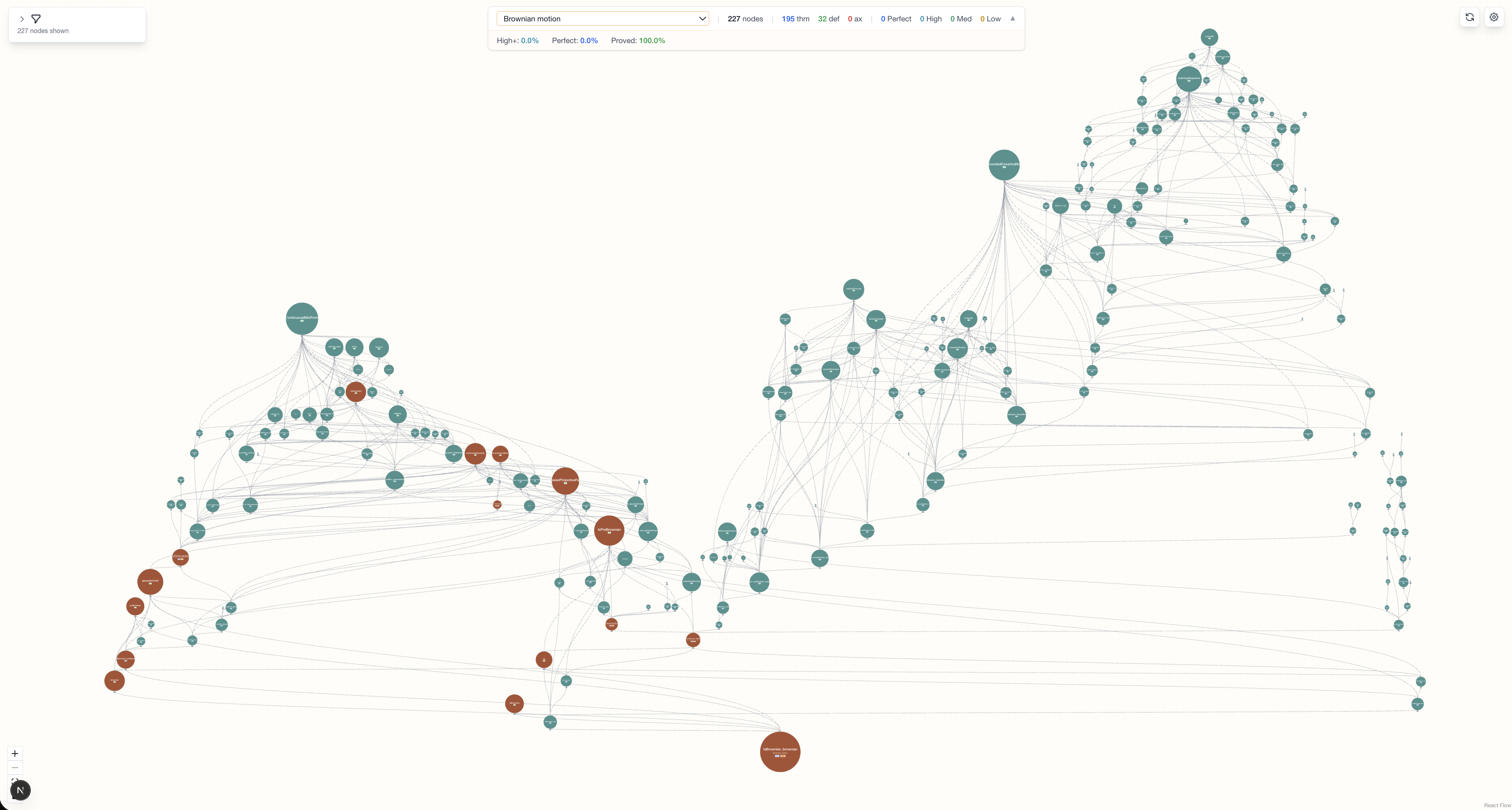}
  \vspace{-0.5em}
  \caption{\small Lean Atlas web viewer. Visualizing the review cone (227 nodes) of the main theorem \texttt{IsBrownian\_brownian} in the Brownian Motion project. The orange nodes (14) are the nodes automatically extracted by Lean Compass as targets for semantic verification (93.8\% reduction).}
  \label{fig:teaser}
\end{figure}

\section{Introduction}
\label{sec:introduction}

AI-driven autoformalization and theorem proving are producing formal proofs at increasing scale~\citep{Hubert2025,Ren2025,Wang2025,Lin2025,Chen2025,ByteDance2025,Baba2025,Varambally2025}, raising new quality assurance challenges~\citep[for a survey, see][]{Yang2024}. Yet a fundamental gap remains: the type checker guarantees logical correctness---that a proof term is correctly constructed for a given proposition---but it does not verify whether the proposition or its definitions represent the intended mathematical content~\citep{LeanDocs}.

This gap gives rise to what we call \textbf{semantic hallucination}: an AI-generated formalization passes the type checker yet does not match the meaning of the original mathematical statement. For instance, when formalizing ``$3/2 = 1.5$'', if the AI omits a type annotation, Lean's type inference interprets the expression as \texttt{Nat}, producing the proposition \texttt{3 / 2 = 1}---a statement that passes the type checker and can be proved, but does not match the original mathematics.

For human review to remain tractable, a formalization workflow needs to answer a concrete question: given a target theorem, which declarations can still affect whether that theorem says the right thing? \textbf{Lean Compass} addresses this question by starting from a selected theorem set and pruning dependencies that arise only from theorem proofs, while retaining dependencies that can still change statement meaning through types or definitions. The result is a smaller project-specific review set for human inspection. Throughout this paper, we treat the Lean standard library and Mathlib as a trusted base.

\begin{figure}[t]
  \centering
  \includegraphics[width=\textwidth]{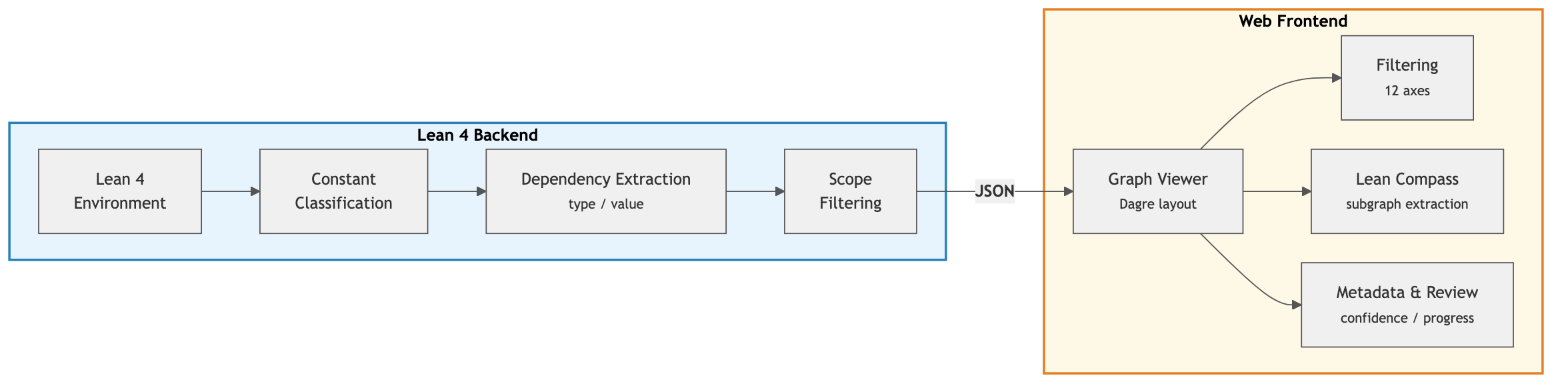}
  \caption{Architecture of Lean Atlas. The Lean~4 backend extracts and classifies the dependency graph, and the web frontend provides interactive visualization. Lean Compass automatically extracts, for a selected theorem set, the project-specific nodes whose semantic correctness can affect those target statements.}
  \label{fig:architecture}
\end{figure}

We place this idea in a \textbf{human-in-the-loop approach}: human scientists and AI collaboratively produce formal proofs, with humans verifying the semantic correctness of propositions and definitions. To realize this workflow, we develop \textbf{Lean Atlas} (Figure~\ref{fig:teaser}), a Lean~4 tool that classifies the dependency graph into type dependencies (proposition- and definition-level) and value dependencies (proof-level), visualizes the result as an interactive web viewer, and uses Lean Compass to focus semantic review on the declarations that remain relevant to the selected targets.

We call formalization code whose propositions and definitions have been semantically verified by a human scientist \textbf{aligned Lean code}---code that carries both logical correctness from the type checker and semantic correctness from human verification.

Although our approach is motivated by mathematical formalization, the semantic hallucination problem is not unique to mathematics. Formal verification with proof assistants is increasingly applied to cryptography, program correctness, theoretical physics, and other scientific domains. Indeed, \citet{ToobySmith2026} recently identified an error in a widely-cited theoretical physics paper through Lean formalization, demonstrating that semantic correctness concerns are not confined to mathematics. When AI provers assist formalization in these areas, the same gap between type-checker validity and intended semantics arises---and may be amplified by domain-specific definitions unfamiliar to proof-assistant experts. Lean Atlas and Lean Compass provide a domain-agnostic framework for narrowing down the nodes requiring human semantic review in any Lean~4 project; our evaluation includes a theoretical physics formalization (PhysLib) and a cryptographic formalization (XMSS Encoding Scheme) as such examples.

The contributions of this paper are threefold.

\begin{enumerate}
  \item We propose a human-in-the-loop approach to the semantic hallucination problem and introduce the concept of \textbf{aligned Lean code}.
  \item We develop \textbf{Lean Atlas}, a tool that classifies each edge of a Lean~4 project's dependency graph into 8 kinds along 3 axes (source kind $\times$ dependency site $\times$ target kind) and provides interactive visualization.
  \item We propose and implement \textbf{Lean Compass}, an algorithm that automatically extracts, for a target theorem set, the project-specific nodes whose semantic correctness can affect those target statements.
\end{enumerate}

\section{Related Work}
\label{sec:related}

Autoformalization---the automatic conversion of natural-language mathematics into formal proof code---has been actively studied since \citet{Wu2022}, who analyzed 150 failure cases and reported that definition mismatches are the most common failure mode. \citet{Li2024} introduced semantic consistency to address meaning breakdowns that symbolic equivalence overlooks. \citet{Liu2025} showed the limitations of type-checking-centric evaluation, \citet{Lu2025} demonstrated that logical validity and BLEU scores fail to detect semantic misalignment, and \citet{Poiroux2025} argued for semantic equivalence as the criterion for judging formalization correctness. We systematically define this problem as \textbf{semantic hallucination} and focus on prevention and scalable human verification workflows rather than detection.

In large-scale Lean formalization projects, leanblueprint~\citep{Massot2020} generates dependency graphs from \LaTeX{} documents for tracking formalization progress. LeanArchitect~\citep{Zhu2026} extracts blueprint metadata directly from Lean code, automating synchronization with \LaTeX{} and providing an interface for integrating AI provers into blueprint workflows. Lean Atlas shares dependency graph visualization with these tools but differs in purpose: it supports semantic verification through the distinction between type and value dependencies, and Lean Compass's automatic extraction of verification targets is a feature unique to Lean Atlas. From a human-AI collaboration perspective, Lean Copilot~\citep{Song2024} and LeanDojo~\citep{Yang2023} provide tactic-level proof assistance but do not address semantic correctness. Similarly, doc-gen4 and miniCTX~\citep{Hu2025} serve documentation and analysis purposes without interactive filtering for semantic verification.

\section{Background}
\label{sec:background}

\subsection{Semantic Hallucination}
\label{subsec:semantic-hallucination}

We assume familiarity with Lean~4~\citep{deMoura2021}. The type checker guarantees logical correctness of proofs but not whether propositions represent the intended mathematical content.

\begin{defn}[Semantic Hallucination]
\label{def:semantic-hallucination}
Given a natural-language mathematical statement $S$ and its formalization $F$, we say that $F$ exhibits semantic hallucination if (1)~$F$ passes the proof assistant's type checker, (2)~$F$ may even have a completed proof, but (3)~$F$ is not semantically equivalent to the mathematical content intended by~$S$.
\end{defn}

The major patterns of semantic hallucination reported in prior work are:

\begin{enumerate}
  \item \textbf{Definition mismatch}~\citep{Wu2022}.
  \item \textbf{Missing or extra assumptions}.
  \item \textbf{Goal substitution}.
  \item \textbf{Quantifier and scope errors}.
  \item \textbf{Type default semantics shift}.
\end{enumerate}

These patterns are not merely theoretical. \citet{Alexeev2025} reported numerous misformalizations encountered while formalizing Erd\H{o}s problems in Lean, including variable/boundary mix-ups (Pattern~4), missing implicit assumptions (Pattern~2), and high-level reformulations producing different problems (Pattern~3).

\subsection{Type Dependency and Value Dependency}

Each constant in Lean~4 has a type (for a theorem, the proposition; for a definition, the type signature) and a value (for a theorem, the proof term; for a definition, the implementation).

\begin{defn}[Type Dependency]
\label{def:type-dependency}
A type dependency from constant $A$ to constant $B$ exists when $B$ appears in the construction of $A$'s type.
\end{defn}

\begin{defn}[Value Dependency]
\label{def:value-dependency}
A value dependency from constant $A$ to constant $B$ exists when $B$ appears only in the construction of $A$'s value (and not in its type).
\end{defn}

A type dependency represents a proposition- or definition-level relationship that a human should verify. The semantic significance of a value dependency depends on the source kind: a value dependency from a theorem represents a proof-level dependency guaranteed by the type checker, whereas a value dependency from a definition can contain computational content not present in the type signature and should be retained for verification. This asymmetry forms the technical foundation of Lean Compass (Section~\ref{sec:lean-compass}).

\section{Lean Atlas}
\label{sec:lean-atlas}

Lean Atlas is an Integrated Proof Environment for supporting scalable semantic verification of AI-generated formalizations. It has a Lean~4 backend that extracts constants and their dependencies, classifies them into type/value dependencies, and exports them in JSON format (implemented as a CLI integrated with Lake), and a web frontend based on Next.js and React Flow for interactive visualization (Figure~\ref{fig:architecture}).

\textbf{Dependency graph extraction.} The backend traverses all constants in the Lean~4 environment. Each constant is classified as theorem, definition, inductive, structure, abbreviation, or axiom. For each constant, its type and value are recursively traversed; dependencies appearing in the type are type dependencies and those appearing only in the value are value dependencies. Each edge is further classified into 8 kinds along the 3 axes of source kind (theorem/definition) $\times$ dependency site (type/value) $\times$ target kind (theorem/definition). Only project-specific constants are included; the Lean standard library and Mathlib are excluded.

\textbf{Interactive web viewer.} The viewer provides 12 independent filtering axes (including kind, confidence, sorry status, and edge kind) combined by AND composition. Metadata (confidence in semantic correctness, proof progress, and definition progress) can be attached to each constant via Lean~4 custom attributes. When the user selects a main theorem, only its transitive dependencies are displayed in a hierarchical layout. Users can update confidence through the viewer, enabling team-based verification tracking.

\section{Lean Compass}
\label{sec:lean-compass}

\subsection{Motivation}

Large Lean~4 projects generate dependency graphs with thousands of nodes. As discussed in Section~\ref{sec:background}, the key insight is the asymmetry in value dependencies: those from a theorem's proof are guaranteed by the type checker and can be pruned, whereas those from a definition's implementation should be retained. Lean Compass exploits this asymmetry to automatically extract a subgraph containing the project-specific nodes whose semantic correctness can affect a target theorem set.

\subsection{Algorithm}

\textbf{Input.} Dependency graph $G = (V, E)$ where each node has a kind and each edge carries source kind, dependency site, and target kind information; target theorem node set $M \subseteq V$. The graph $G$ consists of project-specific nodes; the Lean standard library and Mathlib are treated as a trusted base.

\textbf{Declaration kind aggregation.} Lean Compass aggregates the 6 backend kinds into 3 categories: theorem remains theorem; inductive, structure, and abbreviation are treated as definition; axiom is always included in the output.

\textbf{8-kind edge classification.} Each edge is classified along source kind (theorem/definition) $\times$ dependency site (type/value) $\times$ target kind (theorem/definition), yielding 8 kinds (Table~\ref{tab:edge-kinds}).

\begin{table}[t]
  \centering
  \caption{8-kind edge classification and pruning rules. Edges where the source is a theorem and the site is value (\#3, \#4) are pruned.}
  \label{tab:edge-kinds}
  \small
  \begin{tabular}{clllll}
    \toprule
    \# & EdgeKind & Source & Site & Target & Pruned \\
    \midrule
    1 & \texttt{thm\_type\_to\_def}     & theorem    & type  & definition &            \\
    2 & \texttt{thm\_type\_to\_thm}        & theorem    & type  & theorem    &            \\
    3 & \texttt{thm\_value\_to\_def}    & theorem    & value & definition & \textbf{Pruned} \\
    4 & \texttt{thm\_value\_to\_thm}       & theorem    & value & theorem    & \textbf{Pruned} \\
    5 & \texttt{def\_type\_to\_def}  & definition & type  & definition &            \\
    6 & \texttt{def\_type\_to\_thm}     & definition & type  & theorem    &            \\
    7 & \texttt{def\_value\_to\_def} & definition & value & definition &            \\
    8 & \texttt{def\_value\_to\_thm}    & definition & value & theorem    &            \\
    \bottomrule
  \end{tabular}
\end{table}

\textbf{Pruning rule.} An edge $e = (u, v)$ is pruned when $u$ is a theorem and $e$'s site is value---i.e., \texttt{thm\_value\_to\_def} (\#3) and \texttt{thm\_value\_to\_thm} (\#4).

\textbf{Reachability computation.} On the pruned graph $G' = (V, E')$, the set of nodes reachable from $M$ is computed by BFS. Axioms are always included.

\textbf{Output.} The reachable node set $R \cup A$ (where $A$ = axiom nodes). Lean Compass is parameterized by the target theorem set $M$: the singleton case $M = \{m\}$ is used for per-theorem analysis, while a project-level run chooses $M$ as the selected main theorems of the project. The pseudocode is given in Algorithm~\ref{alg:lean-compass}.

\begin{algorithm}[t]
\caption{LeanCompass}
\label{alg:lean-compass}
\small
\begin{algorithmic}[1]
\Function{ShouldTraverse}{edge $e = (u, v)$}
    \If{$u.\text{kind} = \text{theorem}$ \textbf{and} $e.\text{site} = \text{value}$}
        \State \Return \textbf{false} \Comment{Prune: edge from theorem's proof}
    \EndIf
    \State \Return \textbf{true}
\EndFunction
\Statex
\Function{LeanCompass}{$G, M$}
    \State $R \gets \emptyset$ \Comment{Reachable nodes}
    \State $Q \gets M$ \Comment{BFS queue, initialized with target theorems}
    \While{$Q \neq \emptyset$}
        \State $u \gets Q.\text{dequeue}()$
        \If{$u \in R$} \textbf{continue} \EndIf
        \State $R \gets R \cup \{u\}$
        \ForAll{edge $e = (u, v) \in E$}
            \If{\Call{ShouldTraverse}{$e$} \textbf{and} $v \notin R$}
                \State $Q.\text{enqueue}(v)$
            \EndIf
        \EndFor
    \EndWhile
    \State $A \gets \{v \in V \mid v.\text{kind} = \text{axiom}\}$ \Comment{Always include axioms}
    \State \Return $R \cup A$
\EndFunction
\end{algorithmic}
\end{algorithm}

\subsection{Correctness Argument}

\begin{prop}[Soundness of Compass]
\label{prop:compass-soundness}
If each node in $R \cup A$ is semantically correct and the trusted base (Lean standard library and Mathlib) is semantically correct, then the proposition of each target theorem in $M$ correctly represents the intended mathematical content.
\end{prop}

\begin{proof}[Proof sketch]
Pruning edges from a theorem's proof is safe because the type checker verifies proof terms; if a definition affects a theorem's proposition, it is separately captured as a type dependency. Value dependencies from definitions are retained because implementations can contain computational content beyond their type signatures. All type dependencies (\#1, \#2, \#5, \#6) are retained, so dependency relationships necessary for semantic correctness verification are never lost. By these points, all project-specific nodes that can affect the semantics of each target theorem's proposition in $M$ are included in $R \cup A$.
\end{proof}

\begin{remark}[Scope of Compass]
Since reachability on the pruned graph distributes over set union, running Lean Compass on a target set $M$ is equivalent to taking the union of the singleton results for each $m \in M$. Accordingly, the singleton case $M = \{m\}$ supports visualization and per-theorem analysis, while project-level use chooses $M$ as the selected main theorems of a project. If that selected theorem set exhausts the mathematical claims that the project intends to certify, then the theorem-level soundness guarantee above lifts to project-level claim coverage. Nodes pruned by Compass are unnecessary for verifying the semantic correctness of $M$, but Compass does not claim that those pruned nodes are themselves semantically correct. To verify pruned nodes, one should include them in $M$ and re-apply Compass.
\end{remark}

\section{Evaluation}
\label{sec:evaluation}

\subsection{Experimental Setup}

We selected six Lean~4 formalization projects with different structural characteristics: (1)~\textbf{PrimeNumberTheoremAnd}~\citep{Kontorovich2024}, a proof-heavy formalization of the prime number theorem (8 main theorems from the blueprint); (2)~\textbf{Carleson}~\citep{vanDoorn2023}, a proof-heavy collaborative formalization of Carleson's theorem on pointwise convergence of Fourier series (9 main theorems from the blueprint); (3)~\textbf{Brownian Motion}~\citep{Degenne2025}, a proof-heavy formalization of Brownian motion and stochastic calculus (5 main theorems from 6 annotated declarations, excluding 1 with a trivially small 3-node review cone); (4)~\textbf{FLT}~\citep{Buzzard2023}, a formalization of Fermat's Last Theorem with mixed structure, where we report 6 milestone theorems selected from 10 declarations annotated with \texttt{@[formalMeta mainTheorem]} and omit 1 axiom, 2 internal technical lemmas, and 1 auxiliary property from the averaged milestone statistics; (5)~\textbf{PhysLib}~\citep{ToobySmith2025}, a theoretical physics formalization library covering classical mechanics, quantum mechanics, quantum field theory, and general relativity (5 selected theorems from across the library); and (6)~\textbf{XMSS Encoding Scheme}~\citep{Khovratovich2025}, a definition-heavy formalization of hash-based post-quantum digital signature security bounds (5 main theorems), implemented by the authors. Unlike the first four projects, PhysLib and XMSS represent formal verification applied to theoretical physics and cryptography respectively, and are included to evaluate cross-domain applicability. The PrimeNumberTheoremAnd, Carleson, Brownian Motion, and FLT projects are also used in the evaluation of LeanArchitect~\citep{Zhu2026}.

For each project, we built a dependency graph of project-specific constants using Lean Atlas, applied the Lean Compass pruning rule, and measured node reduction rates. Aggregate values are recorded in \texttt{evaluation-data.md}.

\textbf{Metric.} For each main theorem $m \in M$: review cone $RC(m)$ is the set of project-specific nodes reachable from $m$ on the pre-pruning graph; node reduction rate $= 1 - (\text{nodes after Compass} / |RC(m)|)$. We use node reduction as a proxy for the size of the semantic review candidate set, not as a direct measure of human review time.

\subsection{Results}

Table~\ref{tab:results} shows node counts before and after applying Lean Compass. The main pattern is structural rather than project-specific: reduction is high when a review cone is theorem-dominated and lower when it is definition-dominated. This pattern appears both across projects and within individual projects.

\begin{table}[!t]
  \centering
  \caption{Lean Compass results. For each main theorem, the review cone (number of reachable nodes before pruning) and number of nodes after applying Compass are shown.}
  \label{tab:results}
  \footnotesize
  \begin{tabular}{lrrr}
    \toprule
    Project / Main Theorem ($m$) & Review Cone & After Compass & Reduction \\
    \midrule
    \textbf{PrimeNumberTheoremAnd} & & & \\
    \quad \texttt{Erdos392.Solution\_2} & 315 & 1 & 99.7\% \\
    \quad \texttt{MediumPNT} & 309 & 1 & 99.7\% \\
    \quad \texttt{WeakPNT} & 213 & 2 & 99.1\% \\
    \quad \texttt{WeakPNT\_AP} & 239 & 2 & 99.2\% \\
    \quad \texttt{lambda\_pnt} & 233 & 1 & 99.6\% \\
    \quad \texttt{mu\_pnt} & 231 & 1 & 99.6\% \\
    \quad \texttt{pi\_alt'} & 237 & 1 & 99.6\% \\
    \quad \texttt{prime\_between} & 245 & 1 & 99.6\% \\
    \midrule
    \textbf{Carleson} & & & \\
    \quad \texttt{classical\_carleson} & 1963 & 5 & 99.7\% \\
    \quad \texttt{control\_approximation\_effect} & 1944 & 10 & 99.5\% \\
    \quad \texttt{discrete\_carleson} & 1428 & 53 & 96.3\% \\
    \quad \texttt{exceptional\_set\_carleson} & 1959 & 5 & 99.7\% \\
    \quad \texttt{forest\_complement} & 885 & 105 & 88.1\% \\
    \quad \texttt{forest\_operator} & 1028 & 59 & 94.3\% \\
    \quad \texttt{forest\_union} & 1193 & 105 & 91.2\% \\
    \quad \texttt{metric\_carleson} & 1628 & 25 & 98.5\% \\
    \quad \texttt{two\_sided\_metric\_carleson} & 1783 & 25 & 98.6\% \\
    \midrule
    \textbf{Brownian Motion} & & & \\
    \quad \texttt{IsBrownian\_brownian} & 227 & 14 & 93.8\% \\
    \quad \texttt{IsPreBrownian.exists\_continuous\_modification} & 202 & 2 & 99.0\% \\
    \quad \texttt{IsPreBrownian.hasIndepIncrements} & 48 & 3 & 93.8\% \\
    \quad \texttt{IsPreBrownian.isAEKolmogorovProcess} & 56 & 2 & 96.4\% \\
    \quad \texttt{isProjectiveMeasureFamily\_gaussianProjectiveFamily} & 46 & 5 & 89.1\% \\
    \midrule
    \textbf{FLT} & & & \\
    \quad \texttt{FreyPackage.false} & 27 & 2 & 92.6\% \\
    \quad \texttt{FreyPackage.of\_not\_FermatLastTheorem} & 4 & 2 & 50.0\% \\
    \quad \texttt{ker\_RtoT\_le\_nilradical} & 17 & 8 & 52.9\% \\
    \quad \texttt{Mazur\_Frey} & 23 & 16 & 30.4\% \\
    \quad \texttt{Wiles\_Frey} & 25 & 16 & 36.0\% \\
    \quad \texttt{Wiles\_Taylor\_Wiles} & 31 & 1 & 96.8\% \\
    \midrule
    \textbf{PhysLib} & & & \\
    \quad \texttt{CanonicalEnsemble.fluctuation\_dissipation\_theorem\_finite} & 29 & 18 & 37.0\% \\
    \quad \texttt{ClassicalMechanics.euler\_lagrange\_varGradient} & 150 & 27 & 82.0\% \\
    \quad \texttt{FieldSpecification.wicks\_theorem} & 196 & 94 & 52.0\% \\
    \quad \texttt{QM.OneDimension.HarmonicOscillator.eigenfunction\_completeness} & 59 & 12 & 79.0\% \\
    \quad \texttt{lorentzAlgebra.exp\_mem\_restricted\_lorentzGroup} & 337 & 11 & 96.0\% \\
    \midrule
    \textbf{XMSS Encoding Scheme} & & & \\
    \quad \texttt{Constructions.TL1C.tl1c\_lemma5} & 25 & 18 & 28.0\% \\
    \quad \texttt{Constructions.TLFC.tlfc\_lemma4} & 33 & 22 & 33.3\% \\
    \quad \texttt{Constructions.TSL.tsl\_lemma6} & 45 & 22 & 51.1\% \\
    \quad \texttt{Constructions.random\_oracle\_composition} & 9 & 8 & 11.1\% \\
    \quad \texttt{LowerBound.cost\_lower\_bound} & 23 & 20 & 13.0\% \\
    \bottomrule
  \end{tabular}
\end{table}

In theorem-dominated cones, Compass removes most proof-level dependencies. PrimeNumberTheoremAnd shows the clearest case: all 8 targets reduce by 99.1--99.7\%. Carleson exhibits the same behavior at much larger scale, with 88.1--99.7\% reduction (average 96.2\%) even though its review cones range from 885 to 1963 nodes. Brownian Motion likewise remains highly reducible, with 89.1--99.0\% reduction (average 94.4\%). Figure~\ref{fig:compass-comparison} shows a representative Brownian Motion theorem, \texttt{IsBrownian\_brownian}, whose review cone shrinks from 227 nodes to 14 after pruning.

When computational definitions remain semantically relevant, the reduction drops. XMSS is the most definition-heavy benchmark and achieves only 11.1--51.1\% reduction (average 27.3\%), because chains of value dependencies from definitions must be retained. The same effect appears in lower-layer FLT milestones such as \texttt{Mazur\_Frey} (30.4\%) and \texttt{Wiles\_Frey} (36.0\%), and in PhysLib theorems such as \texttt{CanonicalEnsemble.fluctuation\_dissipation\_theorem\_finite} (37.0\%) and \texttt{FieldSpecification.wicks\_theorem} (52.0\%).

The mixed projects are particularly informative because they contain both extremes. FLT ranges from 30.4\% to 96.8\% reduction: upper-layer results such as \texttt{Wiles\_Taylor\_Wiles} (96.8\%) and \texttt{FreyPackage.false} (92.6\%) behave like theorem-heavy projects, whereas lower-layer constructions retain many definitions for Galois representations and modular forms. PhysLib ranges from 37.0\% to 96.0\%, with \texttt{lorentzAlgebra.exp\_mem\_restricted\_lorentzGroup} (96.0\%) matching the proof-heavy mathematical projects despite belonging to a theoretical physics library.

\begin{figure}[t]
  \centering
  \includegraphics[width=\textwidth]{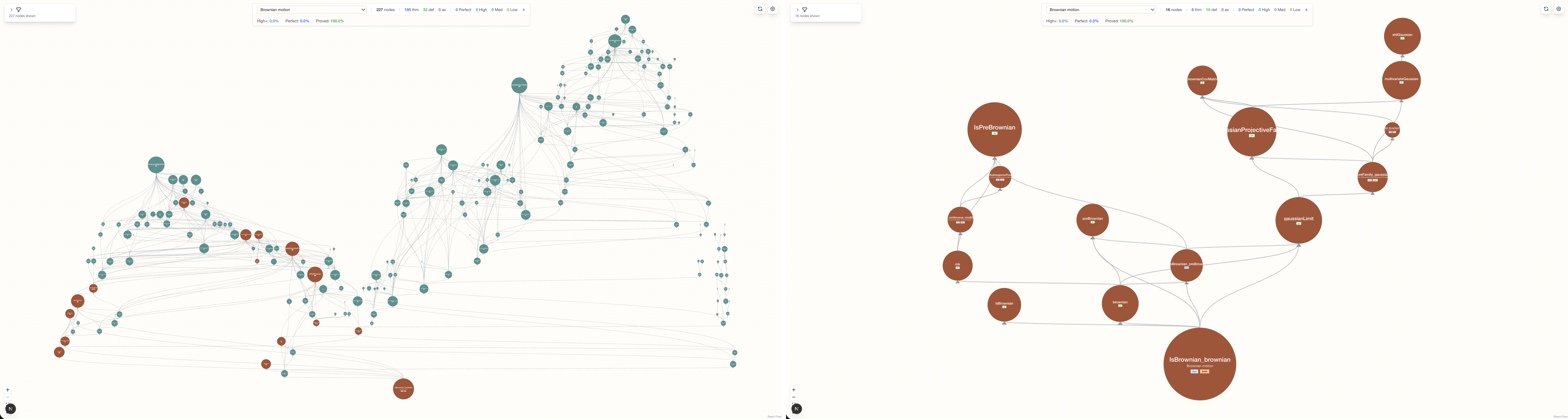}
  \caption{Dependency graph reduction by Lean Compass (Brownian Motion, main theorem \texttt{IsBrownian\_brownian}). (a)~Before: 227 nodes (review cone), (b)~After: 14 nodes (93.8\% reduction).}
  \label{fig:compass-comparison}
\end{figure}

Taken together, these cases indicate that the best predictor of reduction is the theorem/definition ratio inside a review cone, not the project label or even the overall project size. Carleson shows that cones near 2000 nodes can still be pruned aggressively when proof dependencies dominate, while FLT and PhysLib show that a single project can contain both highly reducible and weakly reducible targets. In theorem-dominated cones, value dependencies from proofs are the main removable mass; in definition-dominated cones, many dependencies remain in scope because they contribute directly to the semantics of the target statements.

\section{Conclusion}
\label{sec:conclusion}

We defined semantic hallucination in AI-generated formal proofs and proposed a human-in-the-loop approach realized by Lean Atlas, a Lean~4 tool that classifies dependency graph edges into 8 kinds and provides interactive visualization. Its core algorithm Lean Compass prunes value dependencies from theorem proofs to extract, for a target theorem set, the project-specific nodes whose semantic correctness can affect those target statements. Across six projects, we observed 99\% reduction for proof-heavy PrimeNumberTheoremAnd, 96.2\% for Carleson, 94.4\% for Brownian Motion, 69.0\% for mixed PhysLib, 59.8\% for the reported FLT milestone subset, and 27.3\% for definition-heavy XMSS, with the theorem/definition ratio in each review cone as the determining factor. The inclusion of PhysLib (theoretical physics) and XMSS (cryptography) demonstrates that the approach applies to formal verification in scientific domains beyond pure mathematics.

We introduced \textbf{aligned Lean code}---formalization code carrying both logical correctness from the type checker and semantic correctness from human verification---as a quality standard for AI-generated formalizations. When the selected main theorem set exhausts the mathematical claims that a project intends to certify, the theorem-level guarantee of Lean Compass lifts to project-level claim coverage. Lean Atlas and Lean Compass therefore serve as infrastructure for aligned-code workflows under explicit human semantic review.

The implementation is available at \url{https://github.com/NyxFoundation/lean-atlas}.

\subsection*{Limitations}

The current evaluation extends only to node reduction rates as a proxy for semantic review candidate-set size; actual reduction in human verification time has not been measured. Project-level claim coverage also depends on whether the selected main theorem set adequately captures the mathematical claims that the project intends to certify. Metadata such as confidence currently requires manual annotation, limiting scalability.

\subsection*{Future Work}

\begin{enumerate}
  \item \textbf{Accumulation of aligned Lean code.} Systematically accumulating aligned Lean code to build a formalization codebase free of semantic hallucination, valuable as training data for AI and as a reference for mathematicians.
  \item \textbf{Quantitative evaluation of human verification time.} Conducting comparative experiments with mathematicians to quantify the practical effectiveness of Lean Compass.
  \item \textbf{Integration with AI-generated confidence metadata.} Having AI automatically estimate confidence for each node when generating formalizations, enabling prioritization of human verification.
  \item \textbf{Application to broader scientific domains.} Extending evaluations to formal verification projects in theoretical physics, cryptography, program correctness, and other computational sciences to further validate domain-agnostic applicability.
\end{enumerate}



\end{document}